\documentclass[conference]{IEEEtran}
\usepackage{tikz}
\usepackage{graphicx}
\usepackage{citesort}
\usepackage{amsmath, bm}
\usepackage{amsthm}
\usepackage{amssymb, epsfig}
\usepackage{subfigure,algorithmic,float}
\usepackage{lipsum}
\usepackage{textcomp}

\newcommand\blfootnote[1]{%
  \begingroup
  \renewcommand\thefootnote{}\footnote{#1}%
  \addtocounter{footnote}{-1}%
  \endgroup
}
\usepackage{color}
\usepackage{pgfplots}
\pgfplotsset{compat=newest}
\usetikzlibrary{decorations.markings}
\usetikzlibrary{patterns}
\usetikzlibrary{decorations.pathreplacing,angles,quotes}

\theoremstyle{plain}
\newtheorem{thm}{Theorem} % reset theorem numbering for each chapter
\newtheorem{lemma}{Lemma}
\newtheorem{prop}{Proposition}

\theoremstyle{definition}
\begin{document}
\title{Covert Communication in Fading Channels under Channel Uncertainty}
\author{Khurram Shahzad, Xiangyun Zhou, and Shihao Yan\\ Research School of Engineering, The Australian National University, Canberra, ACT, Australia\\
% \\
Email: {\{khurram.shahzad, xiangyun.zhou, shihao.yan\}@anu.edu.au}}
\maketitle

\begin{abstract}
A covert communication system under block fading channels is considered, where users experience uncertainty about their channel knowledge. The transmitter seeks to hide the covert communication to a private user by exploiting a legitimate public communication link, while the warden tries to detect this covert communication by using a radiometer. We derive the exact expression for the radiometer's optimal threshold, which determines the performance limit of the warden's detector. Furthermore, for given transmission outage constraints, the achievable rates for legitimate and covert users are analyzed, while maintaining a specific level of covertness. Our numerical results illustrate how the achievable performance is affected by the channel uncertainty and required level of covertness.
\end{abstract}
\section{Introduction}
Owing to the broadcast nature of wireless communications, their security is of paramount significance, especially when the transmitted information is important and private.
Traditional ways of ensuring secure transmission over the wireless channel incorporate a multitude of encryption techniques, making sure that even after falling in the wrong hands, the integrity of transmitted information remains intact. Physical layer security, on the other hand, minimizes the information obtained by the eavesdroppers by exploiting the varying characteristics of the wireless channel \cite{sean_book}. However, circumstances exist where it is of vital interest that rather than protecting the content of the transmitted message, the communication itself stays undetectable. Such communication requires some form of covertness to make sure that they are undetected by a warden.\blfootnote{This work was supported by the Australian Research Council\textquotesingle s Discovery Projects under Grant DP150103905.}

There has been a recent interest in the information theoretic aspects of undetectable communication, although some practical aspects of such communication have been long studied by the spread spectrum community \cite{ss_simon,ss_torrieri}. Recently, a square root law has been derived, showing that covert and reliable communication is achievable, given that no more than $\mathcal{O}(\sqrt{n})$ bits are transmitted in $n$ channel uses \cite{bash_jsac}. Further work in this regard has considered the extension of this work to binary symmetric channels (BSCs) \cite{jaggi_ISIT_13}, discrete memoryless channels (DMCs) \cite{wang_TIT,bloch_TIT} and multiple access channels (MACs) \cite{bloch_ISIT_16}. The study in \cite{biao_cc} has shown covert communication with a positive rate considering the distribution of noise uncertainty, instead of the worst-case detection performance a Willie.

In the realm of covert communication, majority of the recent work characterizes the achievable covert throughput in the form of a scaling law, but the exact achievable rates have not been quantified. We focus on this aspect in this work. Specifically, we consider hiding the communication to a covert user, using the communication to a legitimate user as a cover. This scenario has not been studied before in the context of wireless networks, though it is very much in line with standard covert operations, where a public action is used
to provide cover for a secret action.  Our main contributions are as follows:
\begin{itemize}
\item {We exploit the uncertainty in channel knowledge under block fading channels to achieve covertness.}
\item {We derive the exact expression for the optimal threshold of warden's detector (radiometer).}
\item {Under the constraints required for gaining covertness, we analyze the feasible rates for given transmission outage constraints of the legitimate and the covert user.}
\end{itemize}

This paper is organized as follows: Section II details our system model, assumptions and the notation used throughout the paper. Section III discusses the warden's approach to detection of covert communication, while Section IV quantifies the achievable performance of our covert communication system. We present numerical results in Section V, and finally the paper is concluded in Section VI.

\section{System Model}
\begin{figure}[t!]
	\centering
	\includegraphics[scale=0.75]{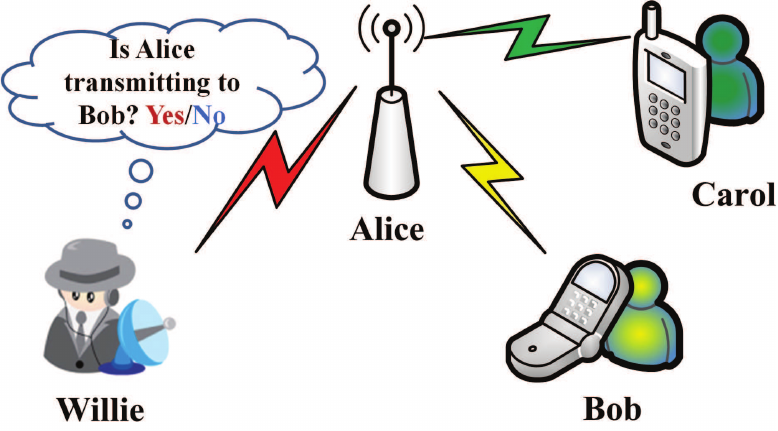}
	\caption{Illustration of the Covert Communication Scenario}
	\label{fig1}
\vspace{-0.3cm}
\end{figure}
We consider a scenario, as shown in Fig. \ref{fig1}, where the transmitter (Alice) openly transmits to the legitimate user (Carol), all the time. Alice also wants to transmit to the covert user (Bob), but she wants to hide this communication from the warden (Willie), using the transmission to Carol as her cover. Willie, being passive, silently observes the communication environment, and tries to detect whether Alice is also transmitting to Bob. It is assumed that Willie knows the transmit power used by Alice, and adopts a radiometer (power detector) as his detector. The distances from Alice-to-Carol, Alice-to-Bob, and Alice-to-Willie are denoted by $d_{ac}$, $d_{ab}$ and $d_{aw}$, respectively, and each user is equipped with a single antenna.

When Alice communicates with Carol or Bob, she transmits her message by mapping it to the sequence $\bm{x}_c=[x^1_c, x^2_c, \dots, x^n_c]$ or $\bm{x}_b=[x^1_b, x^2_b, \dots, x^n_b]$, respectively, where $n$ is the number of channel uses. The average power per symbol in $\bm{x}_c$ and $\bm{x}_b$ is normalized to $1$. Alice employs zero mean Gaussian signalling with variances (i.e., transmit powers) $P_{ac}$ and $P_{ab}$ for Carol and Bob's transmission, respectively. It should be noted here that Alice uses a constant transmit power to Carol, as Carol is unaware of any covert transmission from Alice, and expects a known power at her receiving terminal.

\subsection{Channel Model}
The effect of fading between Alice and user $k$ is modelled by a fading coefficient $h_{ak}$, where $k$ is either $b$ (Bob), $c$ (Carol) or $w$ (Willie). Here $h_{ak}$ follows a circularly symmetric complex Gaussian (CSCG) distribution with zero mean and unit variance, i.e., $h_{ak} \sim \mathcal{CN}(0,1)$. We consider block fading channels, hence the fading coefficients remain constant in one block and change independently from one block to another. We adopt a commonly-used assumption that transmission of a message is completed within one block, i.e., quasi-static fading channels are considered, and the block boundaries are synchronized among all the users. Due to the independent change of fading coefficients among blocks, we focus our analysis on one given block, as the knowledge of previous blocks does not help Willie in improving his detection performance.

While transmitting continuously to Carol, Alice potentially transmits to Bob in a given block. Alice and Bob have a pre-shared secret which enables Bob to know beforehand the block chosen by Alice. Analyzing his observations for a given block, Willie has to decide whether Alice also covertly transmitted to Bob. The null hypothesis $(H_0)$ states that Alice did not talk to Bob, while the alternative hypothesis $(H_1)$ states that Alice did talk to Bob. The signal vector received at user $k$ is
\begin{equation} \label{eq1}
\bm{y}_k = \begin{cases}
\frac{h_{ak}\sqrt{P_{ac}} \bm{x}_c}{d_{ak}^{\alpha/2}} + \frac{h_{ak}\sqrt{P_{ab}}\bm{x}_b}{d_{ak}^{\alpha/2}} +\bm{v}_k, & \text{if $H_1$ is true} \\
\frac{h_{ak}\sqrt{P_{ac}}\bm{x}_c}{d_{ak}^{\alpha/2}} + \bm{v}_k, & \text{if $H_0$ is true}
\end{cases}
\end{equation}
where $\alpha$ is the path-loss exponent, $\bm{v}_k \sim\mathcal{CN}(\bm{0}, \sigma^{2}_{k}\bm{I}_n)$ represents the user $k$'s receiver noise vector, the elements of which follow a CSCG distribution with zero mean and variance $\sigma^{2}_{k}$. Here, $\bm{I}_n$ represents an $n \times n$ identity matrix.

Considering channel uncertainty, the channel coefficient $h_{ak}$ is given by \cite{biao_on_off_13,vakili_error_6}
\begin{equation} \label{eq2}
h_{ak}= {\hat h}_{ak} + \tilde{h}_{ak},
\end{equation}
where ${\hat h}_{ak}$ and $\tilde{h}_{ak}$ represent the known part and the uncertain part of $h_{ak}$ at the corresponding receiver, respectively, and they are zero-mean, independent, CSCG random variables. The variance of the channel uncertainty for user $k$ is denoted by $\beta_k = \mathbb{E}[|\tilde{h}_{ak}|^2]$, $0 \leq \beta_k \leq 1$, and provides a measure of channel uncertainty at user $k$. Accordingly, the variance of $\hat{h}_{ak}$ is $1-\beta_k$, since the variance of $h_{ak}$ is 1.

\subsection{Willie's Detection and Covert Requirement}
Based on his observation vector, $\bm{y}_w$, for one block, Willie has to decide between the hypotheses, $H_0$ and $H_1$, regarding Alice's communication to Bob. Willie knows the complete statistics of its observations under both hypotheses and uses classical hypothesis testing, where the \textit{a priori} probability of either hypothesis being true is $0.5$ \cite{goeckel_fading}. Willie's decision in favor of $H_1$, when $H_0$ is true, is a False Alarm (or type I error), and its probability is denoted by $\mathbb{P}_{FA}$. Similarly, if Willie decides on $H_0$, while $H_1$ is true, then it is a Missed Detection (or type II error), with the probability denoted by $\mathbb{P}_{MD}$. We consider Alice achieving covert communication if, for any $\epsilon > 0$, a communication scheme exists so that $\mathbb{P}_{FA}+\mathbb{P}_{MD} \geq 1-\epsilon$, as $n \rightarrow \infty$ \cite{goeckel_fading}. Here $\epsilon$ signifies the covert requirement, since a sufficiently small $\epsilon$ renders any detector employed at Willie to be ineffective \cite{bash_jsac}.

\section{Detection Strategy at Willie}
From the independent and identically distributed (i.i.d.) nature of Willie's received vector $\bm{y}_w$, given in (\ref{eq1}), each element (symbol) of $\bm{y}_w$ i.e., $y_w^i$ has a distribution given by
\begin{equation} \label{eq3}
%y_w^i \sim
\begin{cases}
\mathcal{CN} (0, |\hat{h}_{aw}|^2\zeta_1 + |\tilde{h}_{aw}|^2 \zeta_1 +\sigma_w^2),  & \text{if $H_1$ is true}            \\
\mathcal{CN} (0, |\hat{h}_{aw}|^2 \zeta_0 + |\tilde{h}_{aw}|^2 \zeta_0 + \sigma_w^2), & \text{if $H_0$ is true}
\end{cases}
\end{equation}
where $\zeta_0 \triangleq \frac{P_{ac}}{d_{aw}^{\alpha}}$ and $\zeta_1 \triangleq \frac{P_{ac}+P_{ab}}{d_{aw}^{\alpha}}$. By application of the Neyman-Pearson criterion, the optimal approach for Willie to minimize his detection error is to use the following \textit{likelihood ratio test} \cite{levy},
\begin{equation} \label{eq4}
\Lambda(\bm{y}_w) = \frac{f_{\bm{y}_w|{\hat h}_{aw}, H_1}(\bm{y}_w|{\hat h}_{aw}, H_1)}{f_{\bm{y}_w|{\hat h}_{aw}, H_0}(\bm{y}_w|{\hat h}_{aw}, H_0)} \underset{D_0}{\overset{D_1}{\gtrless}} \Upsilon ,
\end{equation}
where $\Upsilon = 1$ due to the assumption of equal \textit{a priori} probabilities of each hypothesis. Here, $D_1$ and $D_0$ correspond to a decision in favor of hypothesis $H_1$ and $H_0$, and $f_{\bm{y}_w|{\hat h}_{aw}, H_1}(\bm{y}_w| {\hat h}_{aw}, H_1)$ and $f_{\bm{y}_w|{\hat h}_{aw}, H_0}(\bm{y}_w|{\hat h}_{aw}, H_0)$ are the likelihood functions of Willie's observation vectors for the considered block, under hypothesis $H_1$ and $H_0$, respectively. It should be noted here that the ambiguity at Willie under the two hypotheses comes from two factors, namely, the receiver noise, and the channel uncertainty.

\subsection{Detection using a Radiometer}
We first substantiate that radiometer is indeed the optimal detector for Willie in our system model. We then derive the optimal threshold of the radiometer that minimizes the detection error at Willie.
\begin{lemma}
Under the considered system model, the optimal decision rule that minimizes the detection error at Willie is
\begin{equation} \label{lemma1}
\frac{P_w}{n} \underset{D_0}{\overset{D_1}{\gtrless}} \lambda ,
\end{equation}
which corresponds to a threshold test on $P_w$, where $P_w=\sum_{i=1}^{n}|y_w^i|^2$ is the total power received by Willie in a given block. Here, $\lambda$ is the chosen threshold,  and $n$ is the number of channel uses in a block.
\end{lemma}
\begin{proof}
The proof follows along the same lines as the proof of Lemma $2$ in \cite{goeckel_fading}, where it has been shown using the concepts of stochastic ordering \cite{stoch_order} that a radiometer is optimal for Willie under block fading channels.
\end{proof}

\subsection{Optimal Threshold for Willie's Radiometer}
After establishing the fact that the optimal strategy for Willie is to employ a radiometer, we next evaluate the optimal setting of his radiometer's threshold.
\begin{thm} Using a radiometer for detecting Alice-Bob covert transmission, the optimal value of threshold for Willie's detector is
\begin{equation} \label{eq5}
\begin{aligned}
\lambda^* = \begin{cases}
\lambda^{\dagger}, &\text{if} \quad |\hat{h}_{aw}|^2 < \frac{\lambda^{\dagger} - \sigma_w^2}{\zeta_1}  \\
|\hat{h}_{aw}|^2\zeta_1 + \sigma_w^2, &\text{otherwise}
\end{cases}
\end{aligned}
\end{equation}
where $\lambda^{\dagger} = \frac{\zeta_1 \zeta_0 \beta_w}{\zeta_1 - \zeta_0} \log \left[\frac{\zeta_1}{\zeta_0} \exp \left(\frac{(\zeta_1 - \zeta_0)\sigma_w^2}{\zeta_1 \zeta_0 \beta_w}   \right) \right]$, and $\hat{h}_{aw}$ is Willie's known part of his channel from Alice.
\end{thm}
\begin{proof}
To find the optimal threshold, we consider the optimization problem
\begin{equation} \label{eq6}
\underset{\lambda}{\min} \quad \mathbb{P}_{FA}+\mathbb{P}_{MD}.
\end{equation}
From Lemma $1$, the decision at Willie's detector regarding Alice's transmission to Bob is given by (\ref{lemma1}), where $P_w$ is a sufficient statistic for Willie's detector test. The probabilities of detection error at Willie are given by
\begin{equation}\label{eq8}
\begin{aligned}
\mathbb{P}_{FA} &= \mathbb{P} \left[P_w/n > \lambda \: | \: H_0  \right] \\
&= \mathbb{P}\left[ (\sigma_w^2+|\hat{h}_{aw}|^2\zeta_0+ |\tilde{h}_{aw}|^2\zeta_0)\frac{\chi_{2n}^2}{n} > \lambda  \right],
\end{aligned}
\end{equation}
and
\begin{equation}\label{eq9}
\begin{aligned}
\mathbb{P}_{MD} &= \mathbb{P} \left[P_w/n < \lambda \: | \: H_1  \right]   \\
&= \mathbb{P}\left[ (\sigma_w^2+|\hat{h}_{aw}|^2\zeta_1+ |\tilde{h}_{aw}|^2\zeta_1)\frac{\chi_{2n}^2}{n} < \lambda   \right],
\end{aligned}
\end{equation}
where $\chi_{2n}^{2}$ represents a chi-squared random variable with $2n$ degrees of freedom. From the Strong Law of Large Numbers, we know that $\chi_{2n}^2/n$ converges to $1$, almost surely. The Lebesgue's Dominated Convergence Theorem \cite{math_anlys} allows us to directly replace $\chi_{2n}^2/n$ by $1$, as $n \rightarrow \infty $. Thus for a given realization of $\hat{h}_{aw}$, $\mathbb{P}_{FA}$ and $\mathbb{P}_{MD}$ can be written as
\small
\begin{equation} \label{eq10}
\begin{aligned}
\mathbb{P}_{FA} &= \mathbb{P}\left[ (\sigma_w^2+|\hat{h}_{aw}|^2\zeta_0+ |\tilde{h}_{aw}|^2\zeta_0) > \lambda  \right] \\
&= \mathbb{P}\left[|\tilde{h}_{aw}|^2 > \frac{\lambda - \sigma_w^2 - |\hat{h}_{aw}|^2\zeta_0 }{\zeta_0}   \right] \\
&= \begin{cases}
\exp \left( \frac{|\hat{h}_{aw}|^2\zeta_0 + \sigma_w^2 - \lambda}{\zeta_0\beta_w} \right), & \text{if} \quad  \frac{\lambda - \sigma_w^2 - |\hat{h}_{aw}|^2\zeta_0 }{\zeta_0} \geq 0 \\
1, & \text{otherwise}
\end{cases}
\end{aligned}
\end{equation}
\normalsize
and
\small
\begin{equation} \label{eq11}
\begin{aligned}
\mathbb{P}_{MD} &= \mathbb{P}\left[ (\sigma_w^2+|\hat{h}_{aw}|^2\zeta_1+ |\tilde{h}_{aw}|^2\zeta_1) < \lambda  \right] \\
&= \mathbb{P}\left[|\tilde{h}_{aw}|^2 < \frac{\lambda - \sigma_w^2 - |\hat{h}_{aw}|^2\zeta_1 }{\zeta_1}   \right] \\
&= \begin{cases}
1 - \exp \left( \frac{|\hat{h}_{aw}|^2\zeta_1 + \sigma_w^2 - \lambda}{\zeta_1\beta_w} \right), & \text{if} \quad  \frac{\lambda - \sigma_w^2 - |\hat{h}_{aw}|^2\zeta_1 }{\zeta_1} \geq 0 \\
0, & \text{otherwise}.
\end{cases}
\end{aligned}
\end{equation}
\normalsize
Following (\ref{eq10}) and (\ref{eq11}), we have
\small
\begin{equation} \label{eq12}
\begin{aligned}
 &\mathbb{P}_{FA}+\mathbb{P}_{MD} =
\begin{cases}
1,   & \text{if} \quad \lambda < |\hat{h}_{aw}|^2\zeta_0 + \sigma_w^2 \\
\kappa_0, & \text{if} \quad |\hat{h}_{aw}|^2\zeta_0 + \sigma_w^2 \leq \lambda \leq |\hat{h}_{aw}|^2\zeta_1 + \sigma_w^2 \\
\kappa,  & \text{if} \quad \lambda > |\hat{h}_{aw}|^2\zeta_1 + \sigma_w^2 \\
\end{cases}
\end{aligned}
\end{equation}
\normalsize
where $\kappa = 1-\kappa_1+\kappa_0$, $\kappa_0 \triangleq \exp \left( \frac{|\hat{h}_{aw}|^2\zeta_0 + \sigma_w^2 - \lambda}{\zeta_0\beta_w} \right)$ and $\kappa_1 \triangleq \exp \left( \frac{|\hat{h}_{aw}|^2\zeta_1 + \sigma_w^2 - \lambda}{\zeta_1\beta_w} \right)$. We next analyze the three possible cases in (\ref{eq12}) separately, and find the optimal value of $\lambda$ that minimizes $\mathbb{P}_{FA}+\mathbb{P}_{MD}$.

\subsection*{Case I : $\lambda < |\hat{h}_{aw}|^2\zeta_0 + \sigma_w^2$}
As long as $\lambda < |\hat{h}_{aw}|^2\zeta_0 + \sigma_w^2$, $\mathbb{P}_{FA}+\mathbb{P}_{MD}=1$, and cannot be minimized.

\subsection*{Case II : $|\hat{h}_{aw}|^2\zeta_0 + \sigma_w^2 \leq \lambda \leq |\hat{h}_{aw}|^2\zeta_1 + \sigma_w^2$}
Here, $\mathbb{P}_{FA}+\mathbb{P}_{MD}$ is a decreasing function of $\lambda$, hence Willie chooses the highest possible value of $\lambda$, which is $|\hat{h}_{aw}|^2\zeta_1 + \sigma_w^2$, leading to $\mathbb{P}_{FA}+\mathbb{P}_{MD} = \exp \left(\frac{|\hat{h}_{aw}|^2 (\zeta_0-\zeta_1)}{\zeta_0 \beta_w}\right)$.

\subsection*{Case III : $\lambda > |\hat{h}_{aw}|^2\zeta_1 + \sigma_w^2$}
In order to determine the optimal value of $\lambda$ in this case, we set the first derivative of $\mathbb{P}_{FA}+\mathbb{P}_{MD}$ w.r.t. $\lambda$ equal to zero, which results in
\begin{equation} \label{eq13}
\begin{aligned}
\frac{\partial (\mathbb{P}_{FA}+\mathbb{P}_{MD})}{\partial \lambda} &= \frac{1}{\zeta_1\beta_w} \exp\left(\frac{|\hat{h}_{aw}|^2\zeta_1 + \sigma_w^2 - \lambda}{\zeta_1\beta_w}  \right) \\ &- \frac{1}{\zeta_0 \beta_w} \exp \left( \frac{|\hat{h}_{aw}|^2\zeta_0 + \sigma_w^2 - \lambda}{\zeta_0\beta_w} \right) = 0.
\end{aligned}
\end{equation}
After a few simple manipulations, the optimal value of $\lambda$ in this case is given by
\begin{equation}\label{eq14}
\lambda^{\dagger} \triangleq \frac{\zeta_1 \zeta_0 \beta_w}{\zeta_1 - \zeta_0} \log \left[\frac{\zeta_1}{\zeta_0} \exp \left(\frac{(\zeta_1 - \zeta_0)\sigma_w^2}{\zeta_1 \zeta_0 \beta_w}   \right) \right].
\end{equation}
We note that $\lambda^{\dagger}$ is independent of the channel realization $\hat{h}_{aw}$, and represents the inflection point of $\mathbb{P}_{FA}+\mathbb{P}_{MD}$. It can be verified through simple calculations that $\frac{\partial (\mathbb{P}_{FA}+\mathbb{P}_{MD})}{\partial \lambda} > 0$ for $\lambda > \lambda^{\dagger}$, and $\frac{\partial (\mathbb{P}_{FA}+\mathbb{P}_{MD})}{\partial \lambda} < 0$ for $\lambda < \lambda^{\dagger}$. The second derivative of $\mathbb{P}_{FA}+\mathbb{P}_{MD}$ w.r.t $\lambda$ is
\begin{equation} \label{eq15}
\begin{aligned}
\frac{\partial^2 (\mathbb{P}_{FA}+\mathbb{P}_{MD})}{\partial \lambda^2} = &-\frac{1}{\zeta_1^2\beta_w^2} \exp\left(\frac{|\hat{h}_{aw}|^2\zeta_1 + \sigma_w^2 - \lambda}{\zeta_1\beta_w}  \right) \\ &+ \frac{1}{\zeta_0^2 \beta_w^2} \exp \left( \frac{|\hat{h}_{aw}|^2\zeta_0 + \sigma_w^2 - \lambda}{\zeta_0\beta_w} \right),
\end{aligned}
\end{equation}
which is strictly positive as long as the chosen $\lambda^{\dagger}$ satisfies
\begin{equation}\label{eq16}
\lambda^{\dagger} < \frac{\zeta_1 \zeta_0 \beta_w}{\zeta_1 - \zeta_0} \log \left[\frac{\zeta_1^2}{\zeta_0^2} \exp \left(\frac{(\zeta_1 - \zeta_0)\sigma_w^2}{\zeta_1 \zeta_0 \beta_w} \right)\right],
\end{equation}
 where the requirement in (\ref{eq16}) follows by simply considering the fact that $\zeta_1 > \zeta_0$. Thus $\lambda^{\dagger}$ represents the optimal threshold value for Willie, as long as it satisfies the condition $\lambda^{\dagger} > |\hat{h}_{aw}|^2\zeta_1 + \sigma_w^2$. If $\lambda^{\dagger}$ does not satisfy this, then using the monotonic increase in $\mathbb{P}_{FA}+\mathbb{P}_{MD}$ for $\lambda > \lambda^{\dagger}$, the minimum value of $\lambda$ is chosen that satisfies $\lambda \geq |\hat{h}_{aw}|^2\zeta_1 + \sigma_w^2$.
\end{proof}

\section{Performance of Covert Communication}
Knowing the best detection at Willie, we now consider the overall performance of the covert communication system. We first derive Willie's average detection error probability from Alice's perspective, which will be used to quantify the covertness. Next, we derive the communication outage probabilities at Carol and Bob, which are used to determine the feasible regime of the transmission rates.
\subsection{Average Detection Error Probability}
Using the optimal value of $\lambda$ from (\ref{eq5}), we have
\begin{equation} \label{eq17}
\begin{aligned}
\mathbb{P}_{FA}+\mathbb{P}_{MD} =
\begin{cases}
 1 - \kappa_1^{\dagger} + \kappa_0^{\dagger}, & \text{if} \quad |\hat{h}_{aw}|^2 < \frac{\lambda^{\dagger} - \sigma_w^2}{\zeta_1} \\
 \kappa^{\dagger}, & \text{otherwise}
\end{cases}
\end{aligned}
\end{equation}
where \small $\kappa^{\dagger} \triangleq \exp \left( \frac{|\hat{h}_{aw}|^2(\zeta_0-\zeta_1)}{\zeta_0\beta_w} \right)$, $\kappa_1^{\dagger} \triangleq \exp \left( \frac{|\hat{h}_{aw}|^2\zeta_1 + \sigma_w^2 - \lambda^{\dagger}}{\zeta_1\beta_w} \right)$  \normalsize and \small $\kappa_0^{\dagger} \triangleq \exp \left( \frac{|\hat{h}_{aw}|^2\zeta_0 + \sigma_w^2 - \lambda^{\dagger}}{\zeta_0\beta_w} \right)$. \normalsize Since $\hat{h}_{aw}$ is unknown to Alice, she has to rely on the average measure of Willie's performance to assess the possible covertness. We use $\overline{\mathbb{P}}_{E}^w$ to denote the average $\mathbb{P}_{FA}+\mathbb{P}_{MD}$ over all realizations of $\hat{h}_{aw}$.
\begin{prop}
The average detection error probability at Willie is
\begin{equation} \label{eq18}
\begin{split}
&\overline{\mathbb{P}}_{E}^w = \left[1 - \exp\left(\frac{\sigma_w^2 - \lambda^{\dagger}}{(1-\beta_w)\zeta_1} \right)\right] \\
&\biggl [1-\frac{\beta_w}{2\beta_w -1}\exp\left(\frac{\sigma_w^2 - \lambda^{\dagger}}{\beta_w\zeta_1}\right)
+ \frac{\beta_w}{2\beta_w -1}\exp\left(\frac{\sigma_w^2 - \lambda^{\dagger}}{\beta_w\zeta_0}\right) \biggl ] \\
&+ \exp\left(\frac{\sigma_w^2 - \lambda^{\dagger}}{(1-\beta_w)\zeta_1} \right) \left[\frac{\zeta_0\beta_w}{(1-\beta_w)\zeta_1 + (2\beta_w-1)\zeta_0}   \right].
\end{split}
\end{equation}
\end{prop}
\begin{proof}
From the law of total expectation, we have
\begin{equation}
\begin{aligned}
\overline{\mathbb{P}}_{E}^w = &\mathbb{E}_{|\hat{h}_{aw}|^2} \left[\mathbb{P}_{FA}+\mathbb{P}_{MD}  \right] \\
= &\mathbb{E}_{|\hat{h}_{aw}|^2}\left[\mathbb{P}_{FA}+\mathbb{P}_{MD} \: | \: |\hat{h}_{aw}|^2 < \frac{\lambda^{\dagger} - \sigma_w^2}{\zeta_1}  \right] \\ &\mathbb{P}\left[|\hat{h}_{aw}|^2 < \frac{\lambda^{\dagger} - \sigma_w^2}{\zeta_1}   \right] \\
+ & \:\mathbb{E}_{|\hat{h}_{aw}|^2}\left[\mathbb{P}_{FA}+\mathbb{P}_{MD} \: | \: |\hat{h}_{aw}|^2 \geq \frac{\lambda^{\dagger} - \sigma_w^2}{\zeta_1}  \right] \\ &\mathbb{P}\left[|\hat{h}_{aw}|^2 \geq \frac{\lambda^{\dagger} - \sigma_w^2}{\zeta_1}   \right],
\end{aligned}
\end{equation}
and evaluating this expression completes the proof.
\end{proof}
To achieve covertness, Alice chooses her transmit power levels to Carol and Bob such that.
\begin{equation}
\overline{\mathbb{P}}_{E}^w \geq 1-\epsilon.
\end{equation}

\subsection{Outage Probabilities at Carol and Bob}
\begin{prop}
Under hypothesis $H_1$, the outage probability at Carol for a rate $R_c$ is
\small
\begin{equation} \label{eq_prop1}
\delta_c(H_1) = 1-\frac{P_{\Delta_c}^{\beta_c}}{\beta_c\Delta_c(P_{ac}+P_{ab}) + P_{\Delta_c}^{\beta_c}} \exp \left(-\frac{\Delta_c d_{ac}^{\alpha} \sigma_c^2}{P_{\Delta_c}^{\beta_c}}    \right),
\end{equation}
\normalsize
where $P_{\Delta_c}^{\beta_c} \triangleq (1-\beta_c)\left[P_{ac}-P_{ab}\Delta_c \right] $, and $\Delta_c \triangleq 2^{R_c}-1$.
\end{prop}
\begin{proof}
Under $H_1$, the signal vector received at Carol is
\small
\begin{equation}\label{eq19}
\begin{aligned}
\bm{y}_c = & \: \hat{h}_{ac}\frac{\sqrt{P_{ac}}\bm{x}_c}{d_{ac}^{\alpha/2}}+\tilde{h}_{ac}\frac{\sqrt{P_{ac}}\bm{x}_c}{d_{ac}^{\alpha/2}}
+\hat{h}_{ac}\frac{\sqrt{P_{ab}}\bm{x}_b}{d_{ac}^{\alpha/2}} \\ &+\tilde{h}_{ac}\frac{\sqrt{P_{ab}}\bm{x}_b}{d_{ac}^{\alpha/2}}+\bm{v}_c,
\end{aligned}
\end{equation}
\normalsize
and the signal-to-noise ratio (SNR) is
\begin{equation} \label{eq20}
\begin{aligned}
\text{SNR}_{H_1}^c = \frac{|\hat{h}_{ac}|^2 P_{ac}}{|\hat{h}_{ac}|^2 P_{ab}+|\tilde{h}_{ac}|^2(P_{ac}+P_{ab})+d_{ac}^{\alpha}\sigma_c^2}.
%&= \frac{|\hat{h}_{ac}|^2 \frac{P_{ac}}{d_{ac}^{\alpha}}}{|\hat{h}_{ac}|^2 \frac{P_{ab}}{d_{ac}^{\alpha}}+|\tilde{h}_{ac}|^2\frac{(P_{ac}+P_{ab})}{d_{ac}^{\alpha}}+\sigma_c^2} \\
\end{aligned}
\end{equation}
The outage probability at Carol is
\begin{equation}\label{eq21}
\begin{aligned}
\delta_c(H_1) &= \mathbb{P}\left[ \log_2(1+\text{SNR}_{H_1}^c) < R_c  \right] \\
&= \mathbb{P}\left[ \frac{|\hat{h}_{ac}|^2 P_{ac}}{|\hat{h}_{ac}|^2 P_{ab}+|\tilde{h}_{ac}|^2(P_{ac}+P_{ab})+d_{ac}^{\alpha} \sigma_c^2} < \Delta_c    \right] \\
&= \mathbb{P}\left[|\hat{h}_{ac}|^2 < \frac{\Delta_c \left[ |\tilde{h}_{ac}|^2(P_{ac}+P_{ab})+d_{ac}^{\alpha} \sigma_c^2 \right]}{P_{ac}-P_{ab}\Delta_c} \right],
\end{aligned}
\end{equation}
where $\Delta_c \triangleq 2^{R_c}-1$. Since $\hat{h}_{ac}$ and $\tilde{h}_{ac}$ are independent, thus
\small
\begin{equation} \label{eq22}
\begin{aligned}
\delta_c(H_1) &= \int_{0}^{\infty}  \left[1-\exp \left(- \frac{\Delta_c \left[ |\tilde{h}_{ac}|^2(P_{ac}+P_{ab})+d_{ac}^{\alpha} \sigma_c^2 \right]}{(1-\beta_c)(P_{ac}-P_{ab}\Delta_c)} \right) \right] \\  & \qquad \qquad \qquad \qquad \qquad \qquad \qquad  \cdot f_{|\tilde{h}_{ac}|^2}(|\tilde{h}_{ac}|^2)\mathrm{d}|\tilde{h}_{ac}|^2,
\end{aligned}
\end{equation}
\normalsize
and the solution of this integration gives the desired result.
\end{proof}
It is important to note here that under $H_0$, the outage probability at Carol is
\begin{equation} \label{eq24}
\delta_c(H_0)= 1-\frac{(1-\beta_c)}{\beta_c\Delta_c+(1-\beta_c)}\exp\left(- \frac{\Delta_c d_{ac}^{\alpha} \sigma_c^2}{(1-\beta_c)P_{ac}} \right),
\end{equation}
which has a value lower than $\delta_c(H_1)$ in (\ref{eq_prop1}), due to no interference from Alice-Bob transmission. Thus Carol's performance deteriorates under hypothesis $H_1$.
\begin{prop}
Under hypothesis $H_1$. the outage probability at Bob for a rate $R_b$ is
\small
\begin{equation} \label{eq26}
\begin{aligned}
\delta_b(H_1) = 1-\frac{P_{\Delta_b}^{\beta_b}}{\beta_b\Delta_b(P_{ab}+P_{ac}) + P_{\Delta_b}^{\beta_b} } \exp \left(-\frac{\Delta_b d_{ab}^{\alpha} \sigma_b^2}{P_{\Delta_b}^{\beta_b}}    \right)
\end{aligned}
\end{equation}
\normalsize
where $P_{\Delta_b}^{\beta_b} \triangleq (1-\beta_b)\left[P_{ab}-P_{ac}\Delta_b\right]$, and $\Delta_b \triangleq 2^{R_b}-1$.
\end{prop}
\begin{proof}
The proof follows along the same lines as the proof of Proposition $2$.
\end{proof}
For given outage constraints, e.g. $\delta_c \leq 0.1$ and $\delta_b \leq 0.1$, the achievable rates for Carol and Bob, under $H_1$, can be numerically calculated using (\ref{eq_prop1}) and (\ref{eq26}). For Carol, any achievable rate that satisfies the outage constraint under $H_1$ will naturally satisfy the outage constraint under $H_0$. Hence the focus is on the performance of Carol and Bob under $H_1$.

\section{Numerical Results and Discussion}
In this section, we present the numerical results to show the effect of covertness requirement ($\epsilon$) and channel uncertainty ($\beta$) on the achievable rate region for Carol and Bob. The noise variance of all the users is assumed to be normalized to $1$, and a total transmit power constraint of $30$dB is considered at Alice. For these numerical results, we have considered $\beta_w=\beta_c=\beta_b\triangleq \beta$, while the outage probability constraints at Carol and Bob are $\delta_c \leq 0.1$ and $\delta_b \leq 0.1$, respectively.

\begin{figure}[t!]
\centering
\includegraphics[scale=0.8]{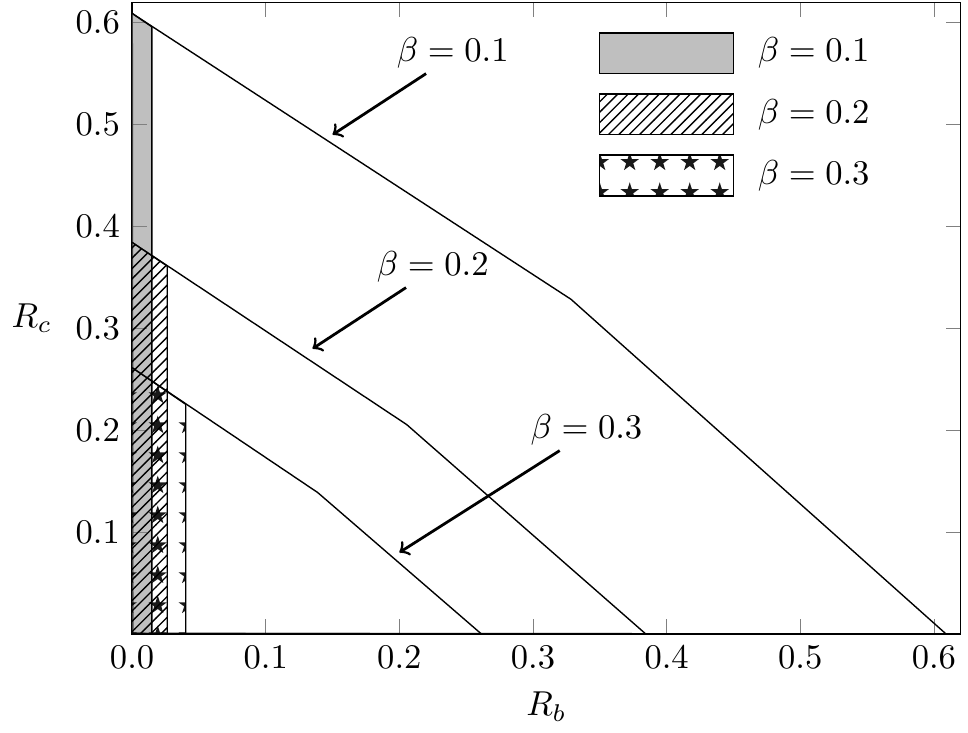}
\caption{The achievable rate region for Carol and Bob under the effect of varying channel uncertainty, $\beta$. Other parameters are $\epsilon=0.2$, $\alpha=3$ and $d_{aw}=d_{ac}=d_{ab}=5$.}
\label{fig2}
\end{figure}

Fig. \ref{fig2} shows the achievable rate region for Carol and Bob, under the effect of changing $\beta$, for a fixed $\epsilon=0.2$. The solid lines, indicated by arrows with varying values of $\beta$, determine the rate region without any covert requirement. It can be observed from the figure that increasing the value of $\beta$ allows Alice to use more power, $P_{ab}$, for transmission to Bob, hence there is an increase in Bob's achievable rate. This increase in feasible $P_{ab}$ is due to the increased channel uncertainty at Willie, causing his detection performance to deteriorate. On the other hand , there is an adverse effect on the overall rate region for Carol and Bob, since the increase in channel uncertainty variance affects their decoding performance. Thus for a fixed covert requirement, increasing the value of $\beta$ in a reasonable range incurs a rate loss for Carol, but increases the achievable rate for Bob.\footnote{It should be noted here that a value of $\beta=0.3$ or $\epsilon=0.3$ is quite large from the practical perspective. We consider such values in our numerical results to illustrate the effect of these parameters on the achievable rate region.}

Fig. \ref{fig3} shows the achievable rates for Carol and Bob, under the effect of changing $\epsilon$, for a fixed $\beta=0.2$. For a fixed channel uncertainty, relaxing $\epsilon$ from $0.1$ to $0.3$ shows an increase in the feasible rate region. Since relaxing $\epsilon$ allows a direct increase in feasible $P_{ab}$ for a given $P_{ac}$, we can clearly see an expansion in the achievable rate region, in favor of Bob.

\begin{figure}[t!]
\centering
\includegraphics[scale=0.8]{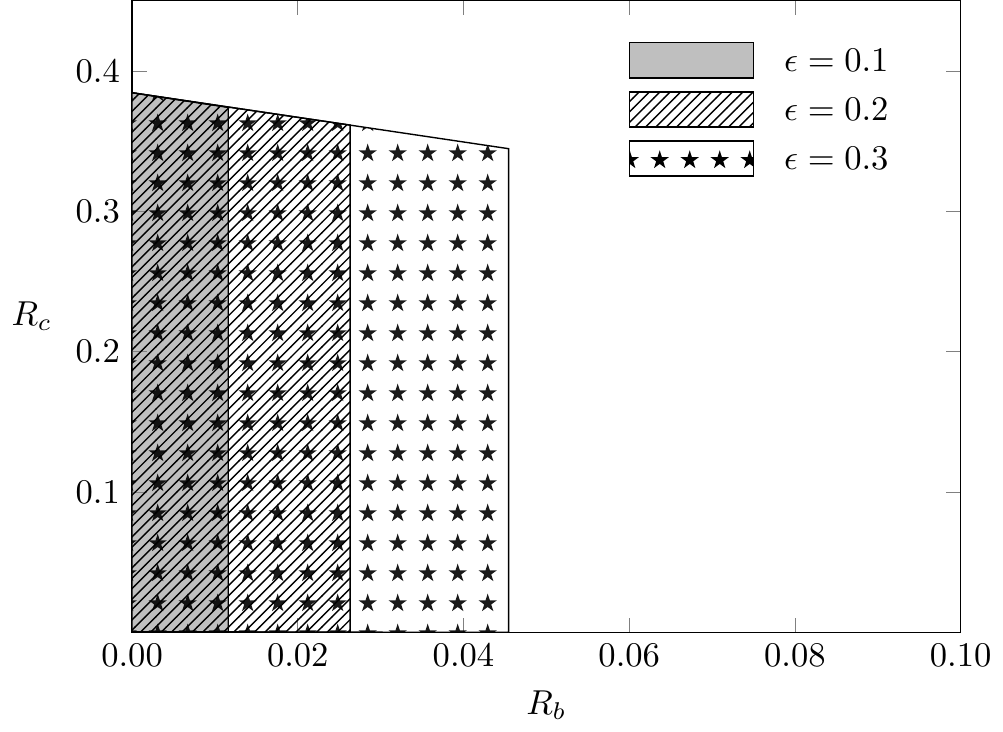}
%\vspace{0.2cm}
\caption{The achievable rate region for Carol and Bob under the effect of varying covertness requirement, $\epsilon$. Other parameters are $\beta=0.2$, $\alpha=3$ and $d_{aw}=d_{ac}=d_{ab}=5$.}
\label{fig3}
\end{figure}

\section{Conclusion}
In this work, we examined how to achieve covert communication in a public and legitimate communication link when users have uncertainty about their channels. We first derived a closed-form expression for the optimal threshold of of Willie’s optimal detector. Next, we quantified the achievable outage rate region for Carol and Bob. Our results showed that the presence of channel uncertainty at Willie allows Alice to achieve a certain amount of covertness, while this channel uncertainty also affects the achievable rates for Carol and Bob.

%\bibliographystyle{IEEETran}
%\bibliography{IEEEabrv,cc_bib_khurram}
%\bibliography{IEEEabrv,cc_bib_khurram}
\end{document}